\documentclass[11pt]{article}
\usepackage{amssymb}
\usepackage{amsmath}
\usepackage{amsthm}
\usepackage[dvips]{graphics,color}
\usepackage{epsfig}
\usepackage{a4wide}
\def\var{\text{Var}}

\def\U{\mathcal{U}}
\def\tr{\textrm{tr}}

\newtheorem{thm}{Theorem}
\newtheorem{lem}{Lemma}

\begin{document}

\title{Real-time covariance estimation for the local level model}

\author{K. Triantafyllopoulos\\ {\it Department of
Probability and Statistics, University of Sheffield, UK} \\ Email: {\tt kostas@sheffield.ac.uk} }

\date{\today}

\maketitle

\begin{abstract}

This paper develops on-line inference for the multivariate local level model, with the focus being placed on covariance estimation of the innovations. We assess the application of the inverse Wishart prior distribution in this context
  and find it too restrictive since the serial correlation structure of the observation
  and state innovations is forced to be the same. We generalize the inverse Wishart distribution to allow for a more convenient correlation structure, but still retaining approximate conjugacy. We prove some relevant results for the new distribution and we develop approximate Bayesian inference, which allows simultaneous forecasting of time series data and estimation of the covariance of the innovations of the model. We provide results on the steady state of the level of the time series, which are deployed to achieve computational savings. Using Monte Carlo experiments, we compare the proposed methodology with existing estimation procedures.  An example with real data consisting of production data from an industrial process is given.

\textit{Some key words:} multivariate time series, covariance estimation, adaptive estimation, dynamic linear models, multivariate control charts.

\end{abstract}

\section{Introduction}\label{model}

Let $\{y_t\}$ be a $p\times 1$ vector process,
generated from the state space model:
\begin{equation}\label{model2}
y_t=\theta_t+\epsilon_t \quad \textrm{and} \quad \theta_t=\phi
\theta_{t-1}+\omega_t,
\end{equation}
where $\theta_t$ is the conditional level of $y_t$, $\phi$ is a scalar
hyperparameter, and the innovation vectors $\epsilon_t$ and
$\omega_t$ follow $p$-variate
Gaussian distributions $\epsilon_t\sim N_p(0,\Sigma)$ and
$\omega_t\sim N_p(0,\Omega)$, for some covariance matrices
$\Sigma$ and $\Omega$, and for some integer $p\geq 1$. It is assumed that the sequences
$\{\epsilon_t\}$ and $\{\omega_t\}$ are individually and mutually
uncorrelated and they are also uncorrelated with the initial state
$\theta_0$, which follows a $p$-variate Gaussian distribution too. For $\phi=1$ the
above model gives the popular local level model, known also as
random walk plus noise model or as steady forecasting model, which is extensively covered in Harvey
(1986, 1989) and in West and Harrison (1997). If $\phi$ lies inside the unit circle, but $\phi\neq 0$, then
(\ref{model2}) can be interpreted as a vector autoregressive model (L\"utkepohl, 2005) with common structure over the component time series. In this paper we focus on the local level model, but the choice $\phi\neq 1$ may allow some small flexibility around it, for example considering nearly local level when $\phi\approx 1$.

Despite its simplicity, the local level model can be used to analyze real data sets in various settings and scenarios, as it has been pointed out by many authors, see e.g. Durbin (2004, p. 6). In the context of model (\ref{model2}) with $\phi\approx 1$, $\theta_t$ is referred to as the conditional level or simply level of $y_t$, since $E(y_t|\theta_t)=\theta_t$ and then $\theta_t$ is local as $\theta_t\approx \theta_{t-1}$, where $E(.)$ denotes expectation. The local level model has been used to analyze the volume of the river Nile (Pole {\it et al.}, 1994, \S7.1; Durbin and Koopman, 2001, \S2.2.2), market research data for a drug development (West and Harrison, 1997, \S2.3), temperature data for assessing global warming (Shumway and Stoffer, 2006, \S6.1), and annual precipitation at Lake Superior (Petris {\it et al.}, 2009, \S3.2.1). A detailed account of the local level model in econometrics, including many examples, is given in Commandeur and Koopman (2007, Chapters 1-7). Furthermore, local level models play a significant role to financial econometrics as they form basic components for unit root tests (Kwiatkowski, 1992). Finally, as pointed out by Triantafyllopoulos (2006), model (\ref{model2}) is a generalization of the Shewhart-Deming model for quality control, and it can be deployed in the context of multivariate control charts for autocorrelated processes (Bersimis {\it et al.}, 2007), where the aim is to signal deviations from the mean and the covariance matrix of these processes.

A central problem associated with inference of model (\ref{model2}), is
the specification or estimation of the covariance matrices $\Sigma$
and $\Omega$. For the estimation of these matrices there are several algorithms based on direct likelihood maximization (Harvey, 1986, 1989) and in particular using analytical score functions (Shephard and Koopman, 1992). Iterative methods for indirect likelihood maximization are also available, e.g. Newton-Raphson algorithms (Shephard and Koopman, 1992; Shumway and Stoffer, 2006, \S6.3) and expectation maximization (EM) algorithms (Dempster {\it et al.}, 1977; Shumway and Stoffer, 1982; Koopman, 1993); Fahrmeir and Tutz (2001, \S8.1.2) and Shumway and Stoffer (2006, \S6.3) have detailed discussions of these algorithms with useful references and recursive versions of the EM algorithm is also possible for on-line application. Simulation based methods, such as Markov chain Monte Carlo (MCMC) (Carter and Kohn, 1994; Gamerman and Lopes, 2006) and sequential Monte Carlo methods (Doucet {\it et al.}, 2001) are also available. Although in recent years the advance in computing power has resulted in sophisticated simulation based and iterative estimation procedures, such as those discussed in the above references, it is still desirable to develop inference that is not based on simulation or on iterative methods, in particular for enabling fast statistical analysis of high dimensional data and sequential model monitoring in real-time (Harrison and West, 1991). The need for real-time estimation has been pointed out in Cantarelis and Johnston (1982) and in many references in machine learning and signal processing, see e.g. Haykin (2001) and Malik (2006). Furthermore, MCMC and maximum likelihood based methods, as those mentioned above, are effectively designed for a ``static'' or in-sample application where a complete set of data is available and the interest is focused on smoothing, rather than on forecasting. Instead, our interest is centred on sequential or ``dynamic'' application, where each time we collect a new observation vector $y_t$ we update the estimates or predictions in an adaptive way.  For the remainder of this paper we discuss approximate conjugate estimation procedures, but in section \ref{examples} we also consider the EM algorithm for comparison purposes.

Assuming that $\Omega$ is proportional to $\Sigma$ in
the sense that $\Omega=w \Sigma$, for some scalar $w>0$,
learning for $\Sigma$ is possible either by adopting Bayesian
methods, considering a Wishart prior for $\Sigma^{-1}$ (West and
Harrison, 1997), or by adopting maximum likelihood estimation
procedures (Harvey, 1986; 1989). The above
proportional structure of $\Sigma$ and $\Omega$ can be seen as a
matrix generalization of the proportionality in the univariate case
$(p=1)$ that leads to the scale observational dynamic model (West
and Harrison, 1997; Triantafyllopoulos and Harrison, 2008), but when
$p>1$ it imposes the restrictive assumption that the correlation
matrix of $\omega_t$ is equal to the correlation matrix of
$\epsilon_t$. This limitation can be understood by noting that the above model belongs to the relatively restricted class of ``seemingly unrelated time series equations'' (SUTSE) (Harvey, 1989, \S8.2), which is a time series extension of the
``seemingly unrelated regression equation models''
(Zellner, 1962). In our opinion, efforts devoted to the estimation of
the above models have been focused primarily on mathematical convenience, and the correlation structure problem mentioned above appears to have been overlooked.

The purpose of this paper is to develop an on-line estimation procedure for adaptive and fast estimation of $\Sigma$ and forecasting
of $\{y_t\}$. The adaptive estimation methods proposed in this paper, may allow for
analysis of high dimensional data, although in this paper this is only briefly explored via Monte Carlo experiments. In order to achieve the above goal we propose the deterministic
specification of $\Omega$ as $\Omega=\Sigma^{1/2}W\Sigma^{1/2}$,
where $\Sigma^{1/2}$ denotes the symmetric square root of $\Sigma$
and $W$ is a covariance matrix to be specified. In our development, $W$ is initially assumed known, but we propose an application of the Newton-Rapshon method for adaptive estimation of this matrix in real problems. We observe that when
$W=wI_p$, where $I_p$ is the $p\times p$ identity matrix, then
$\Omega=w\Sigma$ (leading to $\Omega$ being proportional to
$\Sigma$), but when $W$ is not proportional to $I_p$, then the
correlation matrices of $\epsilon_t$ and $\omega_t$ are different.
Thus we extend the SUTSE models of Harvey (1989, \S8.2) and West and Harrison (1997, \S16.4) to
allow for a more general covariance setting. For estimation purposes,
we deploy approximate Bayesian inference, by adopting a prior
distribution for $\Sigma$ which leads to a generalization of the
inverse Wishart distribution. We provide convergence results of the
posterior covariance matrix of $\theta_t$ leading to the steady
state of $\theta_t$ and this is used in the estimation algorithm in
order to increase its computational speed.

The remaining of the paper is organized as follows. Section \ref{s2s1} generalizes the inverse Wishart distribution and discusses
some properties of the new distribution. In section \ref{s2s2} we develop approximate Bayesian inference for model (\ref{model2}) and section
\ref{examples} includes two illustrations, consisting of Monte Carlo experiments that compare and contrast the performance of our algorithms with existing methods in the literature and an example of monitoring a 5-dimensional process in quality control. Finally, section \ref{discussion} gives concluding comments.

\section{Generalized inverse Wishart distribution}\label{s2s1}

Let $X\sim IW_p(n,A)$ denote that the matrix $X$ follows an
inverse Wishart distribution with $n$ degrees of freedom and with
parameter matrix $A$. Given $A$, we use the notation $|A|$ for the
determinant of $A$, $\textrm{tr}(A)$ for the trace of $A$, and $\textrm{exp}(x)$ for the
exponent of the scalar $x$. The following theorem introduces a generalization of the inverse Wishart distribution.

\begin{thm}\label{lem1}
Consider the $p\times p$ random covariance matrix $X$ and denote
with $X^{1/2}$ the symmetric square root of $X$. Given $p\times p$
covariance matrices $A$ and $S$ and a positive scalar $n>2p$, define
$Y=X^{1/2}S^{-1}X^{1/2}$ so that $Y$ follows an inverse Wishart
distribution $Y\sim IW_p(n,A)$. Then the density function of $X$ is
given by
$$
p(X)=\frac{|A|^{(n-p-1)/2}|S|^{(n-p-1)/2}}{2^{p(n-p-1)/2}\Gamma_p\{(n-p-1)/2\}
|X|^{n/2}}\exp(\tr (-AX^{-1/2}SX^{-1/2}/2)),
$$
where $\Gamma_p(.)$ denotes the multivariate gamma function.
\end{thm}

\begin{proof}
From Olkin and
Rubin (1964) the determinant of the Jacobian matrix of $X$ with
respect to $Y$ is $J(Y\rightarrow X)=J(Y\rightarrow
X^{1/2})J(X^{1/2}\rightarrow X)=\prod_{i\leq
j}^p(\lambda_i+\lambda_j)(\xi_i+\xi_j)^{-1}$, where
$\lambda_1,\ldots,\lambda_p$ are the eigenvalues of
$S^{-1/2}X^{1/2}S^{-1/2}$ and $\xi_1,\ldots,\xi_p$ are the
eigenvalues of $X^{1/2}$. We observe that if $A=I_p$, then $p(X)$ is
an inverse Wishart distribution, since
$\textrm{tr}(-X^{-1/2}SX^{-1/2}/2)=\textrm{tr}(-SX^{-1}/2)$. The
Jacobian $J(Y\rightarrow X)$ does not depend on $A$ and so we can
determine $J(Y\rightarrow X)$ from the special case of $A=I_p$. With
$A=I_p$, $X\sim IW_p(n,S)$ and $Y\sim IW_p(n,I_p)$ and from the
transformation $Y=X^{1/2}S^{-1}X^{1/2}$ we get
$$
p(Y)=\frac{|S|^{(n-p-1)/2}\exp(\textrm{tr}(-Y^{-1}/2))J(X\rightarrow
Y)}{2^{p(n-p-1)/2}\Gamma_p\{(n-p-1)/2\}|S|^{n/2}|Y|^{n/2}}.
$$
Since $Y\sim IW_p(n,I_p)$ it must be
$|S|^{-n/2}|S|^{(n-p-1)/2}J(X\rightarrow Y)=1$ and so
$J(X\rightarrow Y)=|S|^{(p+1)/2}$.

Now, in the general case of a covariance matrix $A$, we see
$$
\int_{X>0} p(X)\,dX=\int_{Y>0}
\frac{|A|^{(n-p-1)/2}}{2^{p(n-p-1)/2}\Gamma_p\{(n-p-1)/2\}|Y|^{n/2}}
\exp(\textrm{tr}(-AY^{-1}/2))\,dY=1,
$$
since $Y\sim IW_p(n,A)$.
\end{proof}

The distribution of the above theorem proposes a generalization of
the inverse Wishart distribution, since if $A=I_p$ we have $X\sim
IW_p(n,S)$ and if $S=I_p$, we have $X\sim IW_p(n,A)$. This is
a new generalization of the
inverse Wishart distribution, differing clearly from the generalizations of Dawid and Lauritzen (1993), Brown
{\it et al.} (1994), Roverato (2002), and Carvalho and West (2007).
In the following we refer to the distribution of Theorem \ref{lem1}
as {\it generalized inverse Wishart} distribution, and we write
$X\sim GIW_p(n,A,S)$. The next result gives some expectations of
the $GIW$ distribution.

\begin{thm}\label{lem2a}
Let $X\sim GIW_p(n,A,S)$ for some known $n,A$ and $S$. Then we have
\begin{enumerate}
\item [(a)] $E(X^{1/2}S^{-1}X^{1/2})=(n-2p-2)^{-1}A$;
$E(X^{-1/2}SX^{-1/2})=(n-p-1)A^{-1}$; \item [(b)]
$E|X|^{\ell}=2^{-p\ell} \prod_{i=1}^p\prod_{j=1}^\ell \{(n-p-i)/2-j\}^{-1} |A|^{\ell}|S|^{\ell}$,
\end{enumerate}
where $E(.)$ denotes expectation and $0<\ell<(n-2p)/2$.
\end{thm}
\begin{proof}
First we prove (a). From the proof of Theorem \ref{lem1} we have that\\
$Y=X^{1/2}S^{-1}X^{1/2}\sim IW_p(n,A)$ and so
$E(Y)=(n-2p-2)^{-1}A$ and $E(Y^{-1})=(n-p-1)A^{-1}$. Proceeding
with (b) we note from the proof of Theorem \ref{lem1} that for any
$n>2p$
$$
\int_{X>0}|X|^{-n/2}\exp(\textrm{tr}(-AX^{-1/2}SX^{-1/2}/2))\,dX=c^{-1},
$$
where $c$ is the normalizing constant of the distribution of $X$.
Then
$$
E|X|^{\ell}=c\int_{X>0}|X|^{-(n-2\ell)/2}\exp(\textrm{tr}(-AX^{-1/2}SX^{-1/2}/2))\,dX=\frac{c}{c^*},
$$
where
$$
c^*=\frac{2^{p\ell}|A|^{(n-p-1)/2}|S|^{(n-p-1)/2}}{2^{p(n-p-1)/2}|A|^{\ell}|S|^{\ell}\Gamma_p\{(n-2\ell-p-1)/2\}}
$$
and the range of $\ell$ makes sure that $n-2\ell>2p$. The result
follows by eliminating the factor $2^{p(n-p-1)/2}$ in the fraction
$c/c^*$, and by noting that from well known properties of the multivariate gamma function we have
$$
\Gamma_p\left(\frac{n-p-1}{2}\right)=\prod_{i=1}^p\Gamma\left(\frac{n-p-i}{2}\right)=\prod_{i=1}^p\prod_{j=1}^\ell \left(\frac{n-p-i}{2}-j\right)\Gamma\left(\frac{n-p-i}{2}-\ell\right)
$$
where $\Gamma(.)$ denotes the gamma function.
\end{proof}
The following property reflects on the symmetry of $A$ and $S$ in
the $GIW$ distribution.
\begin{thm}\label{th:commute}
If $X\sim GIW_p(n,A,S)$, for some known $n$, $A$ and $S$, then
$X\sim GIW_p(n,S,A)$.
\end{thm}
\begin{proof}[Proof of Theorem \ref{th:commute}]
Suppose that $X\sim GIW_p(n,A,S)$. From the normalizing constant of
the density $f(X)$ of Theorem \ref{lem1}, we can exchange the roles
of $|A|$ and $|S|$. And from $\textrm{tr}(-AX^{-1/2}S$ $\times
X^{-1/2}/2) =\textrm{tr}(-SX^{-1/2}AX^{-1/2}/2)$ we have that $X\sim
GIW_p(n,S,A)$.
\end{proof}

Next we show that the mode of a $GIW$ distribution can be obtained by the solution of a
matrix equation. First we give the following lemma.

\begin{lem}\label{lem:trace}
If $X$ is a $p\times p$ real-valued symmetric matrix of variables and $A,S$ are $p\times p$ symmetric
matrices of constants, then
\begin{equation}
\frac{\partial tr(AXSX)}{\partial X} =2\{K+K'-diag(k_{11},\ldots,k_{pp})\}, \label{eq:lem:trace}
\end{equation}
where $K=AXS$ and $K=(k_{ij})_{i,j=1,\ldots,p}$.
\end{lem}
\begin{proof}
Let $X=(x_{ij})_{i,j=1,\ldots,p}$ and thus
$$
\frac{\partial tr(AXSX)}{\partial x_{ij}} = tr\left(A\frac{\partial X}{\partial x_{ij}} SX\right) +
tr\left(AXS \frac{\partial X}{\partial x_{ij}}\right).
$$
Now let $u_j$ be the $j$-th column vector of the identity matrix $I_p$ (a zero vector having one unit in its $j$-th position).
For $i=j$ we have
\begin{equation}
\frac{\partial tr(AXSX)}{\partial x_{ii}} = tr(Au_iu_i'SX) + tr(AXSu_iu_i') = u_i'SXAu_i+u_i'AXSu_i=2k_{ii}, \label{lem:partial1}
\end{equation}
where $u_i'$ denotes the row vector of $u_i$.

For $j<i$ we have
\begin{eqnarray}
\frac{\partial \textrm{tr}(AXSX)}{\partial x_{ij}} &=& \textrm{tr}(SXAu_ju_i') + \textrm{tr}(AXSu_ju_i') + \textrm{tr}(SXAu_iu_j')+\textrm{tr}(AXSu_iu_j') \nonumber \\ &=& 2k_{ji} + 2k_{ij}. \label{lem:partial2}
\end{eqnarray}
Putting together (\ref{lem:partial1}) and (\ref{lem:partial2}) we obtain (\ref{eq:lem:trace}).
\end{proof}

\begin{thm}\label{th:mode}
The mode $\widehat{X}$ of $X\sim GIW(n,A,S)$ satisfies the matrix equation
\begin{equation}\label{mode}
A\widehat{X}^{-1/2}S+S\widehat{X}^{-1/2}A=2n\widehat{X}^{1/2}.
\end{equation}
Furthermore, $\widehat{X}$ is unique, i.e. $GIW$ is a unimodal distribution.
\end{thm}
\begin{proof}
From the density of $X$ we have
$$
\log p(X)=c+n\log |X^{-1/2}|-\frac{1}{2} \textrm{tr}(AX^{-1/2}SX^{-1/2}).
$$
To find $\widehat{X}$ we need to maximize $\log p(X)$. Using Lemma \ref{lem:trace} we have
$$
\frac{\partial \log p(X)}{\partial X^{-1/2}} = n\left\{ 2X^{1/2}-diag\left(x_{11}^{(1/2)},\ldots,x_{pp}^{(1/2)}\right)\right\} - K-K'+diag(k_{11},\ldots,k_{pp}),
$$
where $X^{1/2}=\left\{x_{ij}^{(1/2)}\right\}_{i,j=1,\ldots,p}$, $K=AX^{-1/2}S$ and $K=(k_{ij})_{i,j=1,\ldots,p}$.

Now
$$
\frac{\partial \log p(\widehat{X})}{\partial \widehat{X}^{-1/2}} = 0 \Rightarrow K+K'=2n\widehat{X}^{1/2} \Rightarrow A\widehat{X}^{-1/2}S+S\widehat{X}^{-1/2}A=2n\widehat{X}^{1/2}.
$$
Next we show that at $X=\widehat{X}$ the second partial derivative of $\log p(X)$ is a negative definite matrix.
Let $x=\textrm{vech}(X^{1/2})$ and $x^*=\textrm{vech}(X^{-1/2})$, where $\textrm{vech}(.)$ denotes the vec permutation operator of symmetric matrices. Also, let $D_p$ be the duplication matrix and $H_p$ be any left inverse of it and denote with $\otimes$ the Kronecker product of two matrices. Then
\begin{eqnarray*}
\frac{\partial^2 \log p(X)}{\partial x^* \partial (x^*)'} &=& -n H_p (X^{1/2}\otimes X^{1/2}) D_p - \frac{n}{2} \frac{ \textrm{vech}\{ diag(x_{11}^{(1/2}),\ldots,x_{pp}^{(1/2)})\} }{\partial (x^*)' } \\ &=& -n H_p (X^{1/2}\otimes X^{1/2}) D_p - \frac{n}{2} \frac{ \textrm{vech}\{ diag(x_{11}^{(1/2}),\ldots,x_{pp}^{(1/2)})\} }{\partial x' } \frac{\partial x}{\partial (x^*)'} \\ &=& -n H_p (X^{1/2}\otimes X^{1/2}) D_p + \frac{n}{2} diag\{\textrm{vech}(I_p)\} H_p (X^{1/2}\otimes X^{1/2}) D_p \\ &=& -\frac{n}{2} \left[ 2I_{p(p+1)/2} -diag\{\textrm{vech}(I_p)\}\right] H_p (X^{1/2}\otimes X^{1/2}) D_p <0,
\end{eqnarray*}
since both $2I_{p(p+1)/2} -diag\{\textrm{vech}(I_p)\}$ and $H_p (X^{1/2}\otimes X^{1/2})D_p$ are positive definite matrices.

To prove the second part of the theorem, we write for simplicity $X=\widehat{X}^{1/2}$ so that from (\ref{mode}) we have $AX^{-1}S+SX^{-1}A=2nX$. For $Y \neq X$, with $AY^{-1}S+SY^{-1}A=2nY$, let $\textrm{vec}(.)$ denote the column stacking operator of an unrestricted matrix. Then $A(X^{-1}-Y^{-1})S+S(X^{-1}-Y^{-1})A=2n(X-Y)$ or $Y^{-1}\otimes X^{-1} \textrm{vec}(X-Y)=-2n(S\otimes A+A\otimes S)^{-1}\textrm{vec}(X-Y)$, which leads to the contradiction that $Y^{-1}\otimes X^{-1}$ is a negative definite matrix. Thus $X=Y$ and so the solution $\widehat{X}$ of (\ref{mode}) is unique.
\end{proof}
Some comments are in order.
\begin{enumerate}

\item If $S=I_p$, then $GIW$ is reduced to an inverted Wishart distribution $X\sim IW(n,A)$ with mode $\widehat{X}=A/n$ and this satisfies
equation (\ref{mode}).

\item If $S=\lambda A$, for some $\lambda>0$, then equation (\ref{mode}) can be solved analytically as $\widehat{X}=\lambda n^{-1}A^2$. To see this define $\widehat{Y}^2=\widehat{X}$ so that (\ref{mode}) is satisfied for $\widehat{Y}^{-1}=\sqrt{n/\lambda} A^{-1}$.

\item If $AS$ is symmetric (i.e. if $A$ and $S$ commute), then $\widehat{X}=AS/n$. To prove this first we show that if $A$ and $S$ commute, then $A^{1/2}$ and $S^{1/2}$ commute too. Indeed, assume that $A^{1/2}S^{1/2}\neq S^{1/2}A^{1/2}$, then $AS^{1/2}\neq A^{1/2}S^{1/2}A^{1/2}\neq S^{1/2}A^{1/2}A^{1/2}=S^{1/2}A$ and $AS\neq S^{1/2}AS^{1/2}\neq S^{1/2}S^{1/2}A=SA$, which is a contradiction. Now define $\widehat{Y}^2=\widehat{X}$ with $\widehat{Y}=A^{1/2}S^{1/2}/\sqrt{n}$ and then substitute $\widehat{Y}$ into (\ref{mode}), i.e.
$$
A\widehat{Y}^{-1}S+S\widehat{Y}^{-1}A=\sqrt{n}AA^{-1/2}S^{-1/2}S+\sqrt{n}SS^{-1/2}A^{-1/2}A=2\sqrt{n}A^{1/2}S^{1/2}=2n\widehat{Y}.
$$
Note that the cases (1) and (2) above, are embedded in (3).

\item In the general case we can obtain the solution of (\ref{mode}) numerically, by considering it
as a special case of generalized Sylvester matrix equations (Wu {\it et al.}, 2008).

\end{enumerate}

The next result proposes a way to obtain the unique solution of (\ref{mode}), avoiding numerical methods.
\begin{thm}\label{th:mode2}
The mode $\widehat{X}$ of Theorem \ref{th:mode} satisfies $\textrm{vec}(\widehat{X}^{1/2}\otimes \widehat{X}^{1/2})=c$, where
\begin{gather*}
c=2n(b\otimes B +d\otimes D)(b'b\otimes B^2+b'd\otimes BD +d'b\otimes DB+d'd\otimes D^2)^{-1}vec(I_p),
\end{gather*}
for $b=vec(S)$, $B=I_p\otimes A$, $d=vec(A)$ and $D=I_p\otimes S$.
\end{thm}

\begin{proof}
For simplicity we write $X=\widehat{X}^{1/2}$. Then equation (\ref{mode}) becomes $AX^{-1}SX^{-1}+SX^{-1}AX^{-1}=2nI_p$, which by taking the vectorized operator in both sides, can be written as
\begin{equation}\label{th:mode2:eq1}
(I_p\otimes A)(X^{-1}\otimes X^{-1})\textrm{vec}(S)+(I_p\otimes S)(X^{-1}\otimes X^{-1})\textrm{vec}(A)=2n\textrm{vec}(I_p)
\end{equation}
With $b,B,d,D$ as in the theorem and by taking again the vectorized operator in both sides of (\ref{th:mode2:eq1}) we have
$$
(b'\otimes B+d'\otimes D)\textrm{vec}(X^{-1}\otimes X^{-1})=2n\textrm{vec}(I_p)
$$
Now we can see that the $p^2\times p^4$ matrix $b'\otimes B+d'\otimes D$ is of full rank $p^2$ and so the solution of the above system is given by
$$
\textrm{vec}(X^{-1}\otimes X^{-1}) =  2n (b\otimes B+d\otimes D)\{ (b'\otimes B+d'\otimes D)(b\otimes B+d\otimes D)\}^{-1} \textrm{vec}(I_p) =c,
$$
as required.
\end{proof}

In order to find the mode $\widehat{X}$ using Theorem \ref{th:mode2} we follow the next steps: first we calculate $c$, then we extract the matrix $\widehat{X}^{-1/2}\otimes \widehat{X}^{-1/2}$ from $c$, then using the formula $\textrm{vec}(\widehat{X}^{-1})=\widehat{X}^{-1/2}\otimes \widehat{X}^{-1/2}\textrm{vec}(I_p)$ and rearranging again we find $\widehat{X}^{-1}$ and finally by inverting this matrix we obtain $\widehat{X}$.

However, the above method for the computation of the mode may not be efficient for high dimensional data. Even in low dimensions, as the time series problem we consider in the next section has a sequential application, if we want to use the above procedure for the determination of the mode or if we want to solve the matrix equation of Theorem \ref{th:mode2} using numerical methods, we will have to perform these operations at each time $t$. In our experience this is a heavy computational job, even for relatively short time series. In order to circumvent this difficulty we propose instead to use the estimator
\begin{equation}\label{eq:est:tilda}
\widetilde{X}=\frac{AS+SA}{2n},
\end{equation}
which is motivated by noting that for cases (1)-(3) above, we have $\widetilde{X}=\widehat{X}$. Even when $AS\neq SA$, we have $ASA^{-1}\approx S$, the approximation here refers to matrix similarity, meaning that the matrices $ASA^{-1}$ and $S$ have the same determinant, the same trace, the same eigenvalues, and the same spectrum (see Theorem 21.3.1 of Harville, 1997, p. 525). Thus $AS$ and $SA$ can be thought of being close to each other and estimator \eqref{eq:est:tilda} basically suggests considering the average of $AS/n$ and $SA/n$. Moreover, a close look at $AS$ and $SA$ shows that the diagonal elements of $AS$ and $SA$ are the same and that in \eqref{eq:est:tilda} the off-diagonal elements of $\widetilde{X}$ are averages of the off-diagonal elements of $AS/n$ and $SA/n$. When $AS\approx SA$, and for large $n$, the estimator $|\widetilde{X}|$ is close to $E|X|$, which from (b) of Theorem \ref{lem2a} is equal to $\prod_{i=1}^p(n-p-2-i)^{-1}|A||S|$.

It is easy to verify that if $X\sim GIW_p(n,A,S)$ and considering the partition
$$
X=\left[\begin{array}{cc} X_{11} & 0 \\ 0 & X_{22}\end{array}\right], \quad A=\left[\begin{array}{cc} A_{11} & 0 \\ 0 & A_{22}\end{array}\right] \quad \textrm{and} \quad S=\left[\begin{array}{cc} S_{11} & 0 \\ 0 & S_{22}\end{array}\right],
$$
where $X_{11}$, $A_{11}$, $S_{11}$ are $q\times q$ and $X_{22}$, $A_{22}$, $S_{22}$ are $(p-q)\times (p-q)$ covariance matrices, then $X_{11}\sim GIW_q(n,A_{11},S_{11})$ and $X_{22}\sim GIW_{p-q}(n,A_{22},S_{22})$, for $1\leq q\leq p-1$. The verification of this is just by noting that $|X|=|X_{11}||X_{22}|$ and $\textrm{tr}(AX^{-1/2}SX^{-1/2})=\textrm{tr}(A_{11}X_{11}^{-1/2}S_{11}X_{11}^{-1/2})+\textrm{tr}(A_{22}X_{22}^{-1/2}S_{22}X_{22}^{-1/2})$ so that $p(X)=p(X_{11})p(X_{22})$. From the latter it follows that $X_{11}$ and $X_{22}$ are independent. This result has the following interesting consequence. Suppose that $X_{ii}$ are independent $q_i\times q_i$ random covariance matrices, each following an inverse Wishart distribution $X_{ii}\sim IW_{q_i}(n,k_iS_i)$, for some $k_i>0$ and some covariance matrix $S_i$, with $i=1,\ldots,s$. Then the random matrix $X=\textrm{block diag}(X_{11},\ldots,X_{q_s})$ (the block diagonal matrix of $X_{11},\ldots,X_{ss}$) follows the GIW distribution, $X\sim GIW_p(n,A,S)$, where $A=\textrm{block diag}(k_1I_{q_1},\ldots,k_sI_{q_s})$, $S=\textrm{block diag}(S_{11},\ldots,S_{ss})$ and $p=q_1+\cdots+q_s$. In words, the GIW distribution with the above block diagonal structure on $A$ and $S$ is generated from the superposition of independent inverse Wishart matrices. This gives an interpretation of the matrices $A$ and $S$ in GIW as well as it gives a useful model building approach when we wish to consider the superposition of local level models as in West and Harrison (1997, Chapter 6).

It is also easy to verify that if $X\sim GIW_p(n,A,S)$, then the
density of $Y=X^{-1}$ is
$$
p(Y) = \frac{|A|^{(n-p-1)/2} |S|^{(n-p-1)/2}
|Y|^{(n-2p-2)/2}}{2^{p(n-p-1)/2}\Gamma_p\{(n-p-1)/2\} }
\exp(\textrm{tr}(-AY^{1/2}SY^{1/2}/2)).
$$
This distribution generalizes the Wishart distribution; we will say
that $Y$ follows the {\it generalized Wishart} distribution with
$n-p-1$ degrees of freedom, covariance matrices $A^{-1}$ and
$S^{-1}$, and we will write $Y\sim GW_p(n-p-1,A^{-1},S^{-1})$. We can observe that when $A=I_p$ or $S=I_p$, the above density reduces
to a Wishart density. Again our terminology and notation, should not
cause any confusion with other generalizations of the Wishart
distribution, proposed in the literature (Letac and Massam, 2004).

The next theorem is a generalization of the convolution of the
Wishart and multivariate singular beta distributions (Uhlig, 1994; D\'{i}az-Garc\'{i}a and Guti\'{e}rrez, 1997; Srivastava, 2003).
For some integers $m,n$, let the $p\times p$ random matrix $B$ follow the
multivariate singular beta distribution with $m$ and $n$ degrees of
freedom, respectively, writing $B\sim B_p(m/2,n/2)$. The singularity of the beta distribution considered here is due to $n$ being
smaller than $p$, meaning that $I_p-B$ is singular (a similar argument can be stated for the singularity of $B$, if $m<p$), and thus the density of $B$ does not exist under the Lebesgue measure in the space of the $p\times p$ real-valued covariance matrices, but it does exist
under the Steifel manifold. Under this consideration the density function of $B$ is obtained if we replace the determinant of $I_p-B$ (which is zero) by the product of the positive eigenvalues of $I_p-B$; for more details the reader is referred to the above references.

\begin{thm}\label{th:uhlig}
Let $p$ and $n$ be positive integers and let $m>p-1$. Let
$H\sim GW_p(m+n,A,S)$ and $B\sim B_p(m/2,n/2)$ be independent, where
$A$ and $S$ are known covariance matrices. Then
$$
G\equiv \mathcal{U}(H)'B\mathcal{U}(H) \sim GW_p(m,A,S),
$$
where $\mathcal{U}(H)$ denotes the upper triangular matrix of the
Choleski decomposition of $H$.
\end{thm}

In order to prove this theorem, we prove the somewhat more
general result in the following lemma.

\begin{lem}\label{lemma:uhlig}
Let $A_1\sim W_p(m,I_p)$, $A_2=\sum_{j=1}^nY_jY_j'$, with
$Y_j\sim N_p(0,I_p)$ and $H\sim GW_p(m+n,A,S)$, where $A_1$, $Y_j$
$(j=1,\ldots,n)$ and $H$ are independent. Define $C=A_1+A_2$,
$B=\{\U(C)'\}^{-1}A_1\{\U(C)\}^{-1}$, $G=\U(H)'B\U(H)$ and
$D=H^{1/2}AH^{1/2}-G^{1/2}AG^{1/2}$. Then $C,G,D$ are independent and $C\sim W_p(m+n,I_p)$, $G\sim GW_p(m,A,S)$,
$D=\sum_{j=1}^n Z_jZ_j'$, where $Z_j\sim N_p(0,S)$.
\end{lem}

\begin{proof}
The proof mimics the proof of Uhlig (1994) for the Wishart case. Define
$Z_j=\U(H^{1/2}AH^{1/2})'\{\U(C)'\}^{-1}Y_j$ and note that
$D=\sum_{j=1}^nZ_jZ_j'$. From Theorem \ref{lem1} and from Uhlig
(1994), the Jacobian $J(A_1,H,Y_1,\ldots,Y_n\rightarrow
C,G,Z_1,\ldots,Z_n)$ is $|H|^{-n/2}|C|^{n/2}|A|^{-(p+1)/2}$. Then,
the joint density function of $A_1,H,A_2$ can be written as
\begin{gather*}
p(A_1,H,A_2) = p(A_1) p(H) p(A_2) = \left\{ 2^{pm/2}
\Gamma_p(m/2)\right\}^{-1} \exp(\textrm{tr}(-A_1/2)) |A_1|^{(m-p-1)/2} \\
\times \left[ 2^{p(m+n)/2} \Gamma_p\{(m+n)/2\}|A|^{(m+n)/2} |S|^{(m+n)/2}
\right]^{-1}  \exp(\textrm{tr}(-A^{-1}H^{1/2}S^{-1}H^{1/2}/2))
|H|^{(m+n-p-1)/2} \\ \times (2\pi)^{-pn/2} \exp(\textrm{tr}(-A_2/2))
|A|^{-(p+1)}  \\ = \left[
2^{p(m+n)/2} \Gamma_p\{(m+n)/2\}\right]^{-1}
\exp(\textrm{tr}(-C/2))|C|^{(m+n-p-1)/2} \\ \times \left\{ 2^{pm/2}
\Gamma_p(m/2)|A|^{m/2}|S|^{m/2} \right\}^{-1}
\exp(\textrm{tr}(-A^{-1}G^{1/2}S^{-1}G^{1/2}/2)) |G|^{(m-p-1)/2} \\ \times
(2\pi)^{-pn/2} |S|^{-n/2} \exp(\textrm{tr}(-S^{-1}D/2)) |A|^{(n-p-1)/2} =
p(C)p(G)p(D),
\end{gather*}
where $A_1=|C||B|$, $H^{1/2}AH^{1/2}=G^{1/2}AG^{1/2}+D$ and
$|H|=|G|/|B|$ are used.
\end{proof}

\begin{proof}[Proof of Theorem \ref{th:uhlig}]
The proof is immediate from Lemma \ref{lemma:uhlig}, after noticing
that with the definition of the multivariate singular beta
distribution (Uhlig, 1994), $B\sim B_p(m/2,n/2)$.
\end{proof}

\section{Bayesian inference}\label{s2s2}

\subsection{Estimation forward in time}

In this section we consider estimation for model (\ref{model2}). The
prior distributions of $\theta_0|\Sigma$ and $\Sigma$ are chosen to
be Gaussian and generalized inverse Wishart respectively, i.e.
\begin{gather}
\theta_0|\Sigma\sim
N_p(m_0,\Sigma^{1/2}P_0\Sigma^{1/2})\quad\textrm{and}\quad
\Sigma\sim GIW_p(n_0,Q^{-1},S_0),\label{eq8}
\end{gather}
for some known parameters $m_0$, $P_0=p_0I_p$, $n_0>0$ and
$S_0$. $Q$ is the limit of $Q_t=P_{t-1}+W+I_p$, where $P_{t-1}$ is
a covariance matrix. The next result states that the limit of
$\{P_t\}$ (and hence the limit of $\{Q_t\}$) exist and it provides
the value of this limit as a function of $\phi$ and $W$.

\begin{thm}\label{lem:limit}
If $P_{t}=R_{t}(R_{t}+I_p)^{-1}$, with $R_{t}=\phi^2P_{t-1}+W$,
where $W$ is a positive definite matrix and considering the prior
$P_0=p_0I_p$, for a known constant $p_0>0$, it is
$$
P=\lim_{t\rightarrow\infty}P_t=\frac{1}{2\phi^2}\left[
\left\{(W+(1-\phi^2)I_p)^2+4W\right\}^{1/2}-W-(1-\phi^2)I_p\right],
$$
for $\phi\neq 0$ and $P=W (W+I_p)^{-1}$, for $\phi=0$.
\end{thm}

Before we prove this result we give some background on the limit of
covariance matrices. Let $A\geq 0$ denote that the matrix $A$ is
non-negative definite, let $A>0$ denote that the matrix $A$ is
positive definite and let $A>B$ denote that the matrices $A>0$ and
$B>0$ satisfy $A-B>0$. If $A>B$, then $A^{-1}<B^{-1}$ (Horn and
Johnson, 1999). The sequence of symmetric matrices $\{P_t\}$ is said
to be monotonic and bounded if the scalar sequence
$\{\ell'P_t\ell\}$ is monotonic and bounded, for all real-valued vectors $\ell$. If for all $t$ the
matrix $P_t$ is a non-negative definite matrix, then the above
definition implies that $\{P_t\}$ is bounded if there exist matrices
$L$ and $U$ satisfying $L\leq P_t\leq U$ and monotonic if
$P_{t+1}<P_t$ or $P_{t+1}>P_t$, for any $t>t_0$ and $t_0>0$. If $\{P_t\}$ is both
monotonic and bounded, then it is convergent, since the sequence
$\{\ell'P_t\ell\}$ is also monotonic and bounded and so it is
convergent. The following two lemmas are needed in order to prove
the limit of Theorem \ref{lem:limit}.

\begin{lem}\label{lem3}
If $P_{t}=R_{t}(R_{t}+I_p)^{-1}$, with $R_{t}=\phi^2 P_{t-1}+W$,
where $W$ is a positive definite matrix and $\phi$ is a real number,
then the sequence of $p\times p$ positive matrices $\{P_t\}$ is
convergent.
\end{lem}

\begin{proof}
First suppose that $\phi=0$. Then $R_t=W$, for all $t$, and so
$P_t=W (W+I_p)^{-1}$, which of course is convergent.

Suppose now that $\phi\neq 0$. It suffices to prove that $\{P_t\}$
is bounded and monotonic. Clearly, $0\leq P_{t}$ and since
$\phi^2>0$ and $W$ is positive definite $0<P_{t}$, for all $t>0$.
Since $(R_{t}+I_p)^{-1}>0$,
$(R_{t}+I_p-R_{t})(R_{t}+I_p)^{-1}>0\Rightarrow
P_{t}=R_{t}(R_{t}+I_p)^{-1}<I_p$ and so $0<P_{t}<I_p$. For the
monotonicity it suffices to prove that, if $P_{t-1}^{-1}>
P_{t-2}^{-1}$ (equivalent $P_{t-1}^{-1}< P_{t-2}^{-1}$), then
$P_{t}^{-1}> P_{t-1}^{-1}$ (equivalent $P_{t}^{-1}< P_{t-1}^{-1}$).
From $P_{t-1}^{-1}> P_{t-2}^{-1}$ we have $P_{t-1}<
P_{t-2}\Rightarrow R_{t}< R_{t-1}\Rightarrow R_{t}^{-1}>
R_{t-1}^{-1}\Rightarrow
P_{t}^{-1}-P_{t-1}^{-1}=R_{t}^{-1}-R_{t-1}^{-1}> 0$, since
$P_{t}^{-1}=(R_{t}+I_p)R_{t}^{-1}=I_p+R_{t}^{-1}$. With an analogous
argument we have that, if $P_{t-1}^{-1}< P_{t-2}^{-1}$, then
$P_{t}^{-1}-P_{t-1}^{-1}< 0$, from which the monotonicity follows.
\end{proof}

\begin{lem}\label{lem3post}
Let $\{P_t\}$ be the sequence of Lemma \ref{lem3} and suppose that
$P_0=p_0I_p$, for a known constant $p_0>0$. Then, with $W$ as in
Lemma \ref{lem3}, the limiting matrix
$P=\lim_{t\rightarrow\infty}P_t$ commutes with $W$.
\end{lem}

\begin{proof}
First we prove that if $P_{t-1}$ commutes with $W$, then $P_{t}$
also commutes with $W$. Indeed from $P_{t}=(\phi^2 P_{t-1}+W)(\phi^2
P_{t-1}+W+I_p)^{-1}$ we have that $P_{t}^{-1}=I_p+(\phi^2
P_{t-1}+W)^{-1}$ and then
$$
P_{t}^{-1}W^{-1}=W^{-1}+(\phi^2W P_{t-1}+W^2)^{-1}=W^{-1}+(\phi^2
P_{t-1}W+W^2)^{-1}=W^{-1}P_{t}^{-1}
$$
which implies that $W
P_{t}=(P_{t}^{-1}W^{-1})^{-1}=(W^{-1}P_{t}^{-1})^{-1}=P_{t}W$ and so
$P_{t}$ and $W$ commute. Because $P_0=p_0I_p$, $P_0$ commutes with
$W$ and so by induction it follows that the sequence of matrices
$\{P_t, t\geq 0\}$ commutes with $W$. Since
$P=\lim_{t\rightarrow\infty}P_t$ exists (Lemma \ref{lem3}) we have
$$
PW=\lim_{t\rightarrow\infty}(P_tW)=\lim_{t\rightarrow\infty}(W
P_t)=W P
$$
and so $P$ commutes with $W$.
\end{proof}

\begin{proof}[Proof of Theorem \ref{lem:limit}]
If $\phi=0$, then from Lemma \ref{lem3} we have $P_t=P=W(W+I_p)^{-1}$.

Let $\phi\neq 0$; from Lemma \ref{lem3} we have that $P$ exists and from Lemma
\ref{lem3post} we have that $P$ and $W$ commute. From
$P_{t}=(\phi^2P_{t-1}+W)(P_{t-1}+W+I_p)^{-1}$ we have
$P=(\phi^2P+W)(\phi^2P+W+I_p)^{-1}$ from which we get the equation
$P^2+\phi^{-2}P(W+I_p-\phi^2I_p)-\phi^{-2}W=0$. Now since $P$ and
$W$ commute we can write
\begin{gather*}
P^2+\phi^{-2}P(W+I_p-\phi^2I_p)-\phi^{-2}W=0 \Rightarrow
P^2+\frac{1}{2\phi^2}P(W+(1-\phi^2)I_p) \\
+\frac{1}{2\phi^2}(W+(1-\phi^2)I_p)
P+\frac{1}{4\phi^4}(W+(1-\phi^2)I_p)^2-
\frac{1}{4\phi^4}(W+(1-\phi^2)I_p)^2-W=0\\
\Rightarrow \left(
P+\frac{1}{2\phi^2}(W+(1-\phi^2)I_p)\right)^2=\frac{1}{4\phi^4}(W+(1-\phi^2)I_p)^2+W
\\ \Rightarrow P=\frac{1}{2\phi^2}\left[
\left\{(W+(1-\phi^2)I_p)^2+4W\right\}^{1/2}-W-(1-\phi^2)I_p\right],
\end{gather*}
after rejecting the negative definite root.
\end{proof}

Theorem \ref{lem:limit} generalizes relevant convergence results
for the univariate random walk plus noise model (Anderson and
Moore, 1979, p. 77; Harvey, 1989, p. 119). Following a similar
argument as in Harvey (1989, p. 119), we can see that the speed of
convergence is exponential; for a related discussion on the rate
of convergence the reader is referred to Chan {\it et al.} (1984).

Let $Y\sim t_p(n,m,P)$ denote that the $p$-dimensional random
vector $Y$ follows a multivariate Student $t$ distribution with $n$
degrees of freedom, mean $m$ and scale or spread matrix $P$ (Gupta
and Nagar, 1999, Chapter 4). Let $y^t=(y_1,\ldots,y_t)$ be the information set at time $t$, comprising
data up to time $t$, for $t=1,2,\ldots$ The next result gives an approximate
Bayesian algorithm for the posterior distributions of $\theta_t$
and $\Sigma$ as well as for the one-step forecast distribution of
$y_t$.

\begin{thm}\label{th3}
In the local level model (\ref{model2}), let the initial priors for
$\theta_0|\Sigma$ and $\Sigma$ be specified as in equation
(\ref{eq8}). The one-step
forecast and posterior distributions are approximately given, for
each $1\leq t\leq N$, as follows:
\begin{enumerate}
\item [(a)] One-step forecast at time $t$: $y_t|y^{t-1}\sim
t_p(n_{t-1},\phi m_{t-1},S_{t-1})$, where $S_{t-1}$,
$m_{t-1}$ are known at time $t-1$. \item [(b)] Posteriors at time $t$:
$\theta_{t}|\Sigma,y^{t}\sim
N_p(m_{t},\Sigma^{1/2}P_{t}\Sigma^{1/2})$ and
$\Sigma_{t}|y^{t}\sim GIW(n_t+2p,Q^{-1},S_{t})$, with
$m_{t}=\phi m_{t-1}+A_{t}e_{t}$, $P_{t}=(\phi^2P_{t-1}+W)(\phi^2P_{t-1}+W+I_p)^{-1}$,
$e_{t}=y_{t}-\phi m_{t-1}$, $S_{t}=S_{t-1}+e_{t}e_{t}'$,
where $A_{t}=\Sigma^{1/2}P_{t}\Sigma^{-1/2}$ is approximated
by $A_{t}=\widetilde{\Sigma}^{1/2} P \widetilde{\Sigma}^{-1/2}$, with
$\widetilde{\Sigma}$ the estimator of
$\Sigma_{t}|y^{t}$ as in (\ref{eq:est:tilda}), and
$Q_t=P_{t-1}+W+I_p$ being approximated by its limit
$Q=P+W+I_p$, where $P$ is given by Theorem \ref{lem:limit}.
\end{enumerate}
\end{thm}

\begin{proof}
The proof is inductive in the distribution of $\Sigma|y^t$. Assume
that given $y^{t-1}$ the distribution of $\Sigma$ is
$\Sigma|y^{t-1}\sim GIW(n_{t-1}+2p,Q^{-1},S_{t-1})$.

From the Kalman filter, conditionally on $\Sigma$, the one-step
forecast density of $y_t$ is
$$
y_t|\Sigma,y^{t-1}\sim
N_p(\phi m_{t-1},\Sigma^{1/2}Q_t\Sigma^{1/2})\approx
N_p(\phi m_{t-1},\Sigma^{1/2}Q\Sigma^{1/2}),
$$
where $m_{t-1}$, $Q_t$ and $Q$ are as in the theorem.

Given $y^{t-1}$ the joint distribution of $y_t$ and $\Sigma$ is
\begin{eqnarray}
p(y_t,\Sigma|y^{t-1})&=&p(y_{t}|\Sigma,y^{t-1})p(\Sigma|y^{t-1})\nonumber\\
&=&c_1\frac{\exp(\textrm{tr}(-Q^{-1}\Sigma^{-1/2}
(e_te_t'+S_{t-1})\Sigma^{-1/2}/2))}{|\Sigma
|^{(n_{t-1}+1+2p)/2}},\label{eq9}
\end{eqnarray}
where
$$
c_1=\frac{|S_{t-1}|^{(n_{t-1}+p-1)/2}}{(2\pi)^{\pi/2}2^{p(n_{t-1}+p-1)}|Q|^{(n_{t-1}+p)/2}
\Gamma_p\{(n_{t-1}+p-1)/2\}}.
$$
The one-step forecast density of $y_t$ is
\begin{eqnarray*}
p(y_t|y^{t-1})&=&\int_{\Sigma >0} p(y_t,\Sigma |y^{t-1})\,d\Sigma \\
&=& c_1 \int_{\Sigma >0} |\Sigma |^{-(n_{t-1}+2p+1)/2}
\exp(\textrm{tr}(-Q^{-1}\Sigma^{-1/2}(e_te_t'+S_{t-1}
)\Sigma^{-1/2}/2)) \,d\Sigma \\ &=& c_1 \frac{2^{p(n_{t-1}+p)/2}
\Gamma_p\{(n_{t-1}+p)/2\} }{ |Q|^{-(n_{t-1}+p)/2}
|e_te_t'+S_{t-1}|^{(n_{t-1}+p)/2} } \\ &=&
\frac{\Gamma_p\{(n_{t-1}+p)/2\} } { \pi^{p/2}
\Gamma_p\{(n_{t-1}+p-1)/2\} } |S_{t-1}|^{(n_{t-1}+p-1)/2}
|e_te_t'+S_{t-1}|^{-(n_{t-1}+p)/2},
\end{eqnarray*}
and so $y_{t}|y^{t-1}\sim t_p (n_{t-1}, \phi m_{t-1}, S_{t-1})$, as
required.

Now we derive the distribution of $\Sigma |y^{t}$. Applying Bayes
theorem we have
$$
p(\Sigma |y^t)=\frac{p(y_t|\Sigma ,y^{t-1})p(\Sigma
|y^{t-1})}{p(y_{t}|y^{t-1})}
$$
and from equation (\ref{eq9}) we have
$$
p(\Sigma |y^t)=c_2|\Sigma
|^{-(n_t+2p)/2}\exp(\textrm{tr}(-Q^{-1}\Sigma^{-1/2}S_t \Sigma^{-1/2}/2))
$$
and $n_t=n_{t-1}+1$, where $S_t$ is as in the theorem and the
proportionality constant is $c_2=c_1/p(y_{t}|y_{t-1})$, not
depending on $\Sigma$. Thus $\Sigma |y^t\sim
GIW_p(n_t+2p,Q^{-1},S_t)$ as required. Conditionally on $\Sigma$,
the distribution of $\theta_t$ follows directly from application of
the Kalman filter and this provides the stated posterior
distribution of $\theta_t$.
\end{proof}

From Theorem \ref{th3}, if $W=I_p$, the posterior distribution of
$\Sigma$ is reduced to an inverse Wishart, i.e. $\Sigma
|y^t\sim IW_p(n_t+2p,Q^{-1}S_{t-1})$, where now $Q$ is a variance.
In this case Theorem \ref{th3} reduces to the well known variance
learning of the random walk plus noise model of West and Harrison
(1997). For the application of Theorem \ref{th3}, one can use any estimator of
$\Sigma_t|y^t$, e.g. its mode $\widehat{\Sigma}$; here, following the motivation of $\widetilde{\Sigma}$ in page \pageref{eq:est:tilda}, we have used $\widetilde{\Sigma}$ for presentation purposes, and
this is the estimator we have used and tested in Section \ref{examples}.

\subsection{Choice of hyperparameters}\label{section:hyperp}

The hyperparameter $\phi$ can be chosen \emph{a priori}, e.g. the
application may require a local level model so that $\phi=1$. The covariance matrix $W$ can be optimized by indirect maximization of the log-likelihood
function, which using the prediction decomposition can be expressed
as
\begin{eqnarray}
\ell(W;y^N) &=& \log p(y_1)+\sum_{t=2}^N \log p(y_t|y^{t-1}) \nonumber \\ &=&
c+\frac{1}{2}\sum_{i=1}^N (c_1+t)\log |S_{t-1}| -\frac{1}{2}\sum_{t=1}^N (c_1+t+1)\log |e_te_t'+S_{t-1}|\nonumber \\ &=& c +\frac{(c_1+1)\log |S_0|}{2} -\frac{(c_1+N+1)\log |S_N|}{2}, \label{logl}
\end{eqnarray}
where $c=\sum_{t=1}^N \log \Gamma_p\{(n_0+p+t-1)/2\} - \sum_{t=1}^N \log \Gamma_p\{(n_0+p+t-2)/2\} - 2^{-1}Np \log \pi $ and $c_1=n_0+p-2$. Maximizing the above likelihood is equivalent of minimizing $\log|S_N|=|S_0+\sum_{t=1}^Ne_te_t'|$, given that $S_0$ is bounded. In this objective function to be minimized, $W$ is involved in $e_t$ via the recursion of $m_t$. To simplify notation we consider $\phi=1$. From Theorem \ref{lem:limit}, $W$ is obtained as a function of $P$ as $W=(I_p-P)^{-1}P^2$, given that $I_p-P$ is non-singular. Thus we propose finding $W$ which maximizes the log-likelihood function, conditional on a value of $\Sigma$; this conditioning is proposed for simplification reasons. Instead of working with $W$, we work with $P$, because from the above relationship, we can calculate $W$ from $P$. Even with these simplifications in place $P$ that minimizes $\log|S_N|$, given $\Sigma$ can not be obtained by direct differentiation. Thus we use a Newton-Raphson method to achieve this.

We start by writing recurrently $m_t$ from Theorem \ref{th3} as
\begin{eqnarray*}
m_t&=&m_{t-1}+\Sigma^{1/2}P\Sigma^{-1/2}(y_t-m_{t-1})=(I_p-\Sigma^{1/2}P\Sigma^{-1/2})m_{t-1}+\Sigma^{1/2}P\Sigma^{-1/2}y_t \\ &=& \Sigma^{1/2}(I_p-P)^t\Sigma^{-1/2}m_0+\sum_{i=0}^{t-1}\Sigma^{1/2}(I_p-P)^iP\Sigma^{-1/2}y_{t-i}
\end{eqnarray*}
and then by ignoring the first term (which is justified if $m_0=0$ or if the eigenvalues of $I_p-P$ lie inside the unit circle), we obtain $\log|S_N|$ as
\begin{eqnarray}
\log|S_N|&=&\log\bigg|S_0+\sum_{t=1}^Ny_ty_t'-\sum_{t=2}^N\sum_{i=0}^{t-2}y_ty_{t-1-i}'\Sigma^{-1/2}P(I_p-P)^i\Sigma^{1/2} \nonumber \\ && -\sum_{t=2}^N\sum_{i=0}^{t-2}\Sigma^{1/2}(I_p-P)^iP\Sigma^{-1/2}y_{t-1-i}y_t' \nonumber \\ && +\sum_{t=2}^N \sum_{i=0}^{t-2} \Sigma^{1/2}(I_p-P)^iP
\Sigma^{-1/2}y_{t-1-i}y_{t-1-i}'\Sigma^{-1/2}P(I_p-P)^i\Sigma^{1/2} \bigg|. \label{sn1}
\end{eqnarray}
Since $W$ or $P$ do not depend on $\Sigma$, we proceed by estimating $P$ independently of $\Sigma$, as if $\Sigma$ were proportional to $I_p$. With this in place, using the chain rule of matrix differentiation (Harville, 2007, \S15.7), we obtain the first partial derivative of \eqref{sn1} as $$
\frac{\partial \log|S_N|}{\partial p_{kl} } =\textrm{tr}\left[\left(\frac{\partial \log |S_N|}{\partial S_N}\right)' \frac{\partial S_N}{\partial p_{kl}}\right] =\textrm{tr}\left[\left(2S_N^{-1}-\textrm{diag}(s_{11,N}^{(-1)}.\ldots,s_{pp,N}^{(-1)})\right)(\Lambda_N+\Lambda_N')\right],
$$
where $\Lambda_N=\sum_{t=2}^N\sum_{i=0}^{t-2}(y_t-(I_p-P)^iPy_{t-1-i})y_{t-1-i}'K_i$, $K_i$ is the partial derivative of $P(I_p-P)^i$, $S_N^{-1}=\big(s_{kl,N}^{(-1)}\big)$, and
$\partial \log|S_N| / \partial p_{kl}$ is the $(k,l)$ element of the matrix derivative $\partial \log|S_N| / \partial P$, for $k,l=1,\ldots,p$. For the calculation of $K_i$ we can see that $K_i=K_{i-1}(I_p-P)-P(I_p-P)^{i-1}K_0$, where $K_0=\partial P/\partial p_{kl}$, which, by defining $u_k$ as the zero $p$-dimensional column vector having a unit in the $k$th place, is equal to $u_ku_k'$ when $k=l$ and it is equal to $u_ku_l'+u_lu_k'$ when $k\neq l$ (Harville, 1997, p. 300). The recursion of $K_i$ follows by using the multiplicative rule of differentiation on the function $P(I_p-P)^i$ and writing $K_i$ as a function of $K_{i-1}$.

For the second derivative we have
$$
\frac{\partial^2 \log|S_N|}{\partial p_{kl}\partial p_{rs}}=\textrm{tr}\left(S_N^{-1}\frac{\partial^2S_N}{\partial p_{kl}\partial p_{rs}}\right) - \textrm{tr}
\left(S_N^{-1}\frac{\partial S_N}{\partial p_{kl}}S_N^{-1}\frac{\partial S_N}{\partial p_{rs}}\right)
$$
From before we know $\partial S_N /\partial p_{kl}=\Lambda_N+\Lambda_N'$ and so $\partial^2S_N/\partial p_{kl}\partial p_{rs}=\partial \Lambda_N/\partial p_{rs}+\partial \Lambda_N'/\partial p_{rs}$. Thus
$$
\frac{\partial\Lambda_N}{\partial p_{rs}}=\sum_{t=2}^N\sum_{i=0}^{t-2}\left[y_t-(I_p-P)^iPy_{t-1-i}\right]y_{t-1-i}'\frac{\partial K_i}{\partial p_{rs}} -\sum_{t=2}^N\sum_{i=0}^{t-2}K_i'y_{t-1-i}y_{t-1-i}'K_i,
$$
where from the recursion of $K_i$ we have
$$
\frac{\partial K_i}{\partial p_{rs}} = \frac{\partial K_{i-1}}{\partial p_{rs}} (I_p-P)-2K_{i-1}K_0
$$
with $\partial K_0/\partial p_{rs}=0$. This completes the first and second partial derivatives of $\log|S_N|$ with respect to elements of $P$. Then the Newton-Raphson method at each time $t$ and for iterations $j=1,2,\ldots$, approximates the true minimum $\hat{P}$ by $\hat{P}^{(j)}$, using the formula
\begin{equation}\label{NR1}
\textrm{vec}(\hat{P}^{(j)}) = \textrm{vec}(\hat{P}^{(j-1)}) + \left( \frac{\partial^2 \log |S_N| }{ \partial \textrm{vec}(P) \partial \textrm{vec}(P)' } \right)^{-1}\bigg|_{P=\hat{P}^{(j-1)}} \frac{\partial \log |S_N|}{\partial \textrm{vec}(P)} \bigg|_{P=\hat{P}^{(j-1)}}
\end{equation}
where $\hat{P}^{(0)}=I_p$ and $\textrm{vec}(\cdot)$ denotes the column stacking operator of an unrestricted matrix. Under some regularity conditions (Shumway and Stoffer, 2006, \S6.3), the algorithm converges to the true minimum $\hat{P}$. Convergence is assumed at iteration $j$, for which $\parallel \hat{P}^{(j)}-\hat{P}^{(j-1)}\parallel_2\leq Tol$, for some small tolerance value $Tol$, where $\parallel \cdot \parallel_2$ denotes the Frobenius norm or distance; similar stoppage rules are discussed in Shumway and Stoffer (2006, \S6.3). Note that typically not many iterations are needed for convergence, although this may depend on the specific application and on the dimension on the data; for the examples in the next section we have used $Tol=0.001$.

An alternative approach is to consider indirect optimization of the conditional log-likelihood function using the Expectation Maximization (EM) algorithm. Some recursive type or on-line version of the EM algorithm is possible (not discussed further in this paper), but if the reader is more familiar with the typical off-line EM algorithm described in Koopman (1993) and Shumway and Stoffer (2006, \S6.3), such an approach would prevent the application of real-time estimation.

Finally, we discuss the specification of $W$ using discount factors (West and Harrison, 1997, Chapter 6). According to this, we introduce $p$ (not necessarily distinct) discount factors $\delta_1,\ldots,\delta_p$, forming a discount matrix $\Delta=\textrm{diag}(\delta_1,\ldots,\delta_p)$. The idea of this specification is that the prior covariance matrix $\var(\theta_t|\Sigma,y^{t-1})=\phi^2\Sigma^{1/2}\Delta^{-1/2}P\Delta^{-1/2}\Sigma^{1/2}$ is increased compared to $\var(\theta_{t-1}|\Sigma,y^{t-1})$, reflecting on the increased uncertainty or loss of information going from $t-1$ to $t$, prior to observing $y_t$. From the above, the expression of $P$ in Theorem \ref{lem:limit} and by equating $\phi^2P+W=\phi^2\Delta^{-1/2}P\Delta^{1/2}$, we obtain the matrix equation $(W^2+4W)^{1/2}+W=\Delta^{-1/2}(W^2+W)^{1/2}-\Delta^{-1/2}W\Delta^{-1/2}$. For $\delta_1=\ldots=\delta_p=\delta$ (known as single discounting), the solution of this equation is $W=\delta^{-1}(1-\delta)^2I_p$, which is proportional to $I_p$, and so, in this case, the GIW distribution reduces to an IW, as discussed in section \ref{s2s1}. In the general case, it can be shown that the solution of the above matrix equation yields $W$ to be diagonal (but not necessarily proportional to $I_p$), i.e. $W=\textrm{diag}(\delta_1^{-1}(1-\delta_1)^2,\ldots,\delta_p^{-1}(1-\delta_p)^2)$. However, still there remains the problem of the specification of the discount factors. A commonly adopted approach, is to include the discount factors to the likelihood function, and to maximize it with respect to them, but this takes us again back to the indirect maximization procedure. In this paper, we favour the Newton-Raphson methodology as described above, but we do recognize its limitations, in particular regarding high dimensional data where the inversion of the Hessian matrix may be difficult or even impossible. In such cases a suitable approach involving discount factors may be favoured.

\subsection{Time-varying covariance matrices}

So far our discussion has been focused on situations where $\Sigma$, the conditional covariance matrix of $y_t$ is time-invariant. However, in many situations, in particular in finance, this is not the case. For example consider that $y_t$ denotes the logarithm of the price of $p$ assets, or the logarithm of $p$ foreign exchange rates. It is evident that model (\ref{model2}) would not be an appropriate model to consider as $\Sigma$, interpreted here as the volatility of $y_t$, should be time-varying. We can thus extend model (\ref{model2}) by replacing $\Sigma$ by a time-varying $\Sigma_t$ and including a stochastic process to describing the evolution of $\Sigma_t$. For the volatility covariance matrix $\Sigma_t$, we propose a
multiplicative stochastic law of its precision $\Sigma_t^{-1}$, i.e.
\begin{equation}\label{evol}
\Sigma_{t}^{-1}=k\U(\Sigma_{t-1}^{-1})'B_t\U(\Sigma_{t-1}^{-1}),\quad
t=1,\ldots,N,
\end{equation}
where $k=\{\delta(1-p)+p\}\{\delta(2-p)+p-1\}^{-1}$, for a discount
factor $0<\delta\leq 1$, and $\U(\Sigma_{t-1}^{-1})$ denotes the unique
upper triangular matrix based on the Choleski decomposition of
$\Sigma_{t-1}^{-1}$. Here $B_t$ is a $p\times p$ random matrix
following the multivariate singular beta distribution
$B_t\sim B(m/2,1/2)$, where $m=\delta(1-\delta)^{-1}+p-1$. The motivation behind the above evolution has been discussed in the literature, see e.g. Uhlig (1994, 1997). Here $k,\delta,m$ are chosen so that a random walk type evolution for $\Sigma_t^{-1}$ is preserved, i.e. $E(\Sigma_t^{-1}|y^{t-1})=E(\Sigma_{t-1}^{-1}|y^{t-1})$. This model is a generalization of Shephard's local scale models (Shephard, 1994), which were suggested as an alternative to integrated GARCH modelling and which are exploiting the gamma/beta convolution proposed by Smith and Miller (1986).

If we combine Theorems \ref{th:uhlig} and \ref{th3} we can obtain the full estimation of the above model; in brief Theorem \ref{th:uhlig} is responsible for the prior estimation or prediction of $\Sigma_t$, given data $y^{t-1}$ and Theorem \ref{th3} is responsible for the posterior estimation of $\Sigma_t$ given $y^t$ and of the estimation of $\theta_t$ and the prediction of $y_t$. Next we give the result, the proof of which is trivial by the discussion above.

\begin{thm}\label{th3:vol}
In the local level model (\ref{model2}) with a time-varying volatility covariance matrix $\Sigma=\Sigma_t$ and evolution (\ref{evol}), let the initial priors for
$\theta_0|\Sigma_0$ and $\Sigma_0$ be specified as in equation
(\ref{eq8}). The one-step
forecast and posterior distributions are approximately given, for
each $1\leq t\leq N$, as follows:
\begin{enumerate}
\item [(a)] One-step forecast at time $t$: $\Sigma_{t}|y^{t-1}\sim GIW_p(\delta
(1-\delta)^{-1}+2p,Q^{-1},k^{-1} S_{t-1})$ and $y_{t}|y^{t-1}\sim
t_p(\delta(1-\delta)^{-1},m_{t-1},k^{-1}S_{t-1})$, where
$k=(\delta(1-p)+p)(\delta(2-p)+p-1)^{-1}$ and $\delta$, $S_{t-1}$,
$m_{t-1}$ are known at time $t-1$. \item [(b)] Posteriors at time $t$:
$\theta_{t}|\Sigma_{t},y^{t}\sim
N_p(m_{t},\Sigma_{t}^{1/2}P_{t}\Sigma_{t}^{1/2})$ and
$\Sigma_{t}|y^{t}\sim GIW((1-\delta)^{-1}+2p,Q^{-1},S_{t})$, with
$m_{t}=\phi m_{t-1}+A_{t}e_{t}$, $P_{t}=(\phi^2P_{t-1}+W)(\phi^2P_{t-1}+W+I_p)^{-1}$,
$e_{t}=y_{t}-\phi m_{t-1}$, $S_{t}=k^{-1}S_{t-1}+e_{t}e_{t}'$,
where $A_{t}=\Sigma_{t}^{1/2}P_{t}\Sigma_{t}^{-1/2}$ is approximated
by $A_{t}=\widetilde{\Sigma}_t^{1/2} P \widetilde{\Sigma}_t^{-1/2}$, with
$\widetilde{\Sigma}_t$ the estimator of
$\Sigma_{t}|y^{t}$ as in (\ref{eq:est:tilda}), and
$Q_t=P_{t-1}+W+I_p$ being approximated by its limit
$Q=P+W+I_p$, where $P$ is given by Theorem \ref{lem:limit}.
\end{enumerate}
\end{thm}

Some comments are in order. First note that if we set $\delta=1$, then $k=1$ and $B_t=1$ with probability 1 and Theorem \ref{th3:vol} is very similar to Theorem \ref{th3}, the only difference being that the finite $n_t$ in Theorem \ref{th3} becomes $\infty$ in Theorem \ref{th3:vol} and this means that the $t$ distribution of $y_t|y^{t-1}$ practically becomes a normal distribution under Theorem \ref{th3:vol}. Another point refers to the suitability of the evolution (\ref{evol}) and the related local level model. Multivariate stochastic volatility models that allow for $y_t$ and/or for $\Sigma_t$ to follow a vector or matrix autoregressive processes have been proposed in the literature (Chib {\it et al.}, 2006; Philipov, 2006; Maasoumi and McAleer, 2006), but they have to rely on simulation-based methods, typically on Markov chain Monte Carlo or on particle filters, and they may not be suitable for real-time prediction of high dimensional data. Such a demand has recently become more and more prevalent as hedge funds and other investment boutiques require reliable automatic forecasting procedures that are suitable for algorithmic statistical arbitrage (Montana {\it et al.}, 2009). In this direction the above algorithm offers an option, which extends a series of papers in this area, see e.g. Quintana and West (1987), Quintana {\it et al.} (2003), Soyer and Tanyeri (2006), Carvalho and West (2007) and references therein.

\section{Illustrations}\label{examples}

In this section we report on Monte Carlo experiments, in order to compare the performance of the proposed algorithm with existing estimation procedures, and also we present an application to multivariate control charting.

\subsection{Monte Carlo experiments}\label{montecarlo}

We have generated realizations of observation and evolution covariance matrices $\Sigma$ and $\Omega$ according to the following scheme: for each covariance matrix, first we generate independently $p(p+1)/2$ correlations from a beta distribution and we multiply them by $+1$ or $-1$ generated by a bernoulli distribution with probability 1/2. Next we generate independently $p$ variances from a gamma distribution, and then we use the correlation decomposition of the covariance matrix, i.e. $\Sigma=V C V$, where $V$ is the diagonal matrix with elements the square roots of the $p$ simulated variances and $C$ is the correlation matrix with off-diagonal elements the $p(p+1)/2$ simulated correlations and with units in the main diagonal.

With this scheme in place we have performed a Monte Carlo study, over a set of 100 simulated $p$-variate time series vectors according to the local level model (\ref{model2}) with $\phi=1$ and for three time series lengths $N=100,500,1000$. We have considered $p=10,50,100$, covering from low to relatively high dimensional time series, and for their estimation we contrast the algorithm of the previous section (this model is referred to as GIW) with a local level model where the observation covariance matrix is estimated via an inverted Wishart distribution (this model is referred to as IW), a local level model where both of $\Sigma$ and $\Omega$ are estimated using the EM algorithm of Shumway and Stoffer (1982) (this model is referred to as EM), and the local level model using the true simulated values of $\Sigma$ and $\Omega$ (this model is referred to as Kalman). For the IW model $\Sigma$ was estimated by assuming an inverse Wishart prior and $\Omega=w\Sigma$, where $w$ was estimated by direct maximum likelihood methods as in Harvey (1986, 1989). For all models we used the priors \eqref{eq8} with $m_0=(0,\ldots,0)'$ and $p_0=1000$, the latter of which reflects on a weakly informative or vague prior specification for $\theta_0$. Also, for both the IW and GIW models we used the prior \eqref{eq8} for $\Sigma$, the difference being that when $W$ is a covariance matrix (for the GIW) this prior becomes $\Sigma\sim GIW_p(n_0,Q^{-1},S_0)$, while when $W=w$ is s scalar variance (for the IW), this prior reduces to $\Sigma\sim IW_p(n_0,Q^{-1}S_0)$; for both cases $n_0=1/100$ and $S_0=I_p$. For the estimation of the GIW, at each time $t$, $W$ is estimated by the Newton-Raphson method of section \ref{section:hyperp}.

Table \ref{table1} reports on the average mean of squared standardized one-step forecast errors (MSSE), which if the fit is perfect should be equal to the unit. Here, due to the high dimensions considered, we witness the quality of the estimation of $\Sigma$ and $\Omega$ via the accuracy of the one-step forecast covariance matrix for each model. First of all we note that the values of the MSSE for the Kalman model are nearly equal to one and clearly this model is the benchmark or the gold standard here, but artificial as in practice we will not know these covariance matrices. We observe that the GIW produces consistent results, outperforming the IW, and producing MSSE close to the gold standard. In comparison with the EM we observe that at low dimensions $p$ and for small values of $N$, the GIW is better, although as $N$ increases the performance of EM is improved and for $N=1000$ the EM model produces marginally better results than the GIW. The improved performance of the EM model at $N=1000$ is expected as it is well known that, under certain conditions, the EM estimators of $\Sigma$ and $\Omega$ converge to their true values (Shumway and Stoffer, 2006, \S6.3). But as $p$ increases we observe a deterioration in the performance of the EM model as compared to the GIW; in particular for $p=100$ and $N=100$ we still obtain reasonable performance with the GIW model, while both the EM and IW models clearly overestimate $\Sigma$ and $\Omega$. Here it should be noted that in our setup both the IW and the EM models are aimed at off-line application, since they need the whole data path $y^N$ for the computation of maximum likelihood estimates. For the EM algorithm we used the convergence criterion used in Shumway and Stoffer (2006, p. 345) that convergence is assumed when the likelihood function does not change by more than 0.001. However, this has resulted in slower algorithms, in particular at the higher dimensions considered here. For a single model, the algorithm of the GIW run in 1 minute and 31 seconds (for $p=10$ and $N=100$) and in 3 minutes and 19 seconds (for $p=100$ and $N=1000$); the respective results for the other models were, for the IW 41 seconds (for $p=10$ and $N=100$) and 1 minute and 43 seconds (for $p=100$ and $N=1000$), for the EM 1 minute and 47 seconds (for $p=10$ and $N=100$) and 3 minutes and 53 seconds (for $p=100$ and $N=1000$), and for the Kalman 11 seconds (for $p=10$ and $N=100$) and 55 seconds (for $p=100$ and $N=1000$). The experiments were run on an Intel(R) Celeron(R) M processor 1.60GHz, 504MB of RAM and the software we used was the freeware R, version 2.9.1, downloadable from {\tt http://www.r-project.org/}.

\begin{table}[t]
\caption{Performance of the algorithm of the previous section (GIW), against the local level models using the inverse Wishart distribution (IW), using estimates of the covariances matrices from the EM algorithm (EM), and using the true covariance matrices (Kalman). Shown are the Monte Carlo averages of the mean of standardized one-step ahead forecast errors for each model and their respective standard errors, given in brackets.}\label{table1}
\begin{center}
\begin{tabular}{|rr|rrrr|}
\hline  $p$ & $N$ & GIW & IW & EM & Kalman \\  \hline  10 & 100 & 0.983 (0.002) & 0.963 (0.003) & 0.972 (0.050) & 0.999 (0.000) \\ & 500 & 0.995 (0.002) & 0.975 (0.003) & 0.996 (0.003) & 0.998 (0.000) \\ & 1000 & 0.997 (0.000) & 0.988 (0.001) & 0.998 (0.001) & 1.001 (0.000) \\ \hline \hline 50 & 100 & 0.969 (0.001) & 0.911 (0.001)  & 1.060 (0.003) & 0.998 (0.001) \\ & 500 & 0.985 (0.004) & 1.045 (0.004) & 1.066 (0.002) & 1.002 (0.000) \\ & 1000 & 1.011 (0.003) & 1.039 (0.001) & 1.009 (0.002) & 1.002 (0.000) \\ \hline \hline 100 & 100 & 0.969 (0.001) & 0.899 (0.001)  & 1.160 (0.002) & 0.995 (0.001) \\ & 500 & 0.972 (0.003) & 1.074 (0.002) & 1.082 (0.001) & 1.003 (0.000) \\ & 1000 & 1.005 (0.001) & 1.032 (0.001) & 1.004 (0.003) & 1.001 (0.000) \\ \hline
\end{tabular}
\end{center}
\end{table}

To mark the quality of the estimation for the GIW model, Figure \ref{fig1a} plots the Frobenius distance of the estimated $\Sigma$ at each time point ($t=1-1000$) from its true simulated value, for $p=10,50,100$. We note that in all three cases the algorithm achieves an upper bound 0.008 quite quickly. The distances of $p=50$ and $p=100$ are much more volatile in comparison to the distance of $p=10$, but all eventually converge. The means of the three distances were 0.0050, 0.0053 and 0.0054 respectively and their respective variances were $6.25\times 10^{-7}$, $2.16\times 10^{-6}$ and $2.60\times 10^{-6}$, respectively. The respective distances of the estimated $W$ follow a similar pattern to that of Figure \ref{fig1a} and their accurate estimation appears to be an important element of the successful estimation of $\Sigma$.

 \begin{figure}[t]
 \epsfig{file=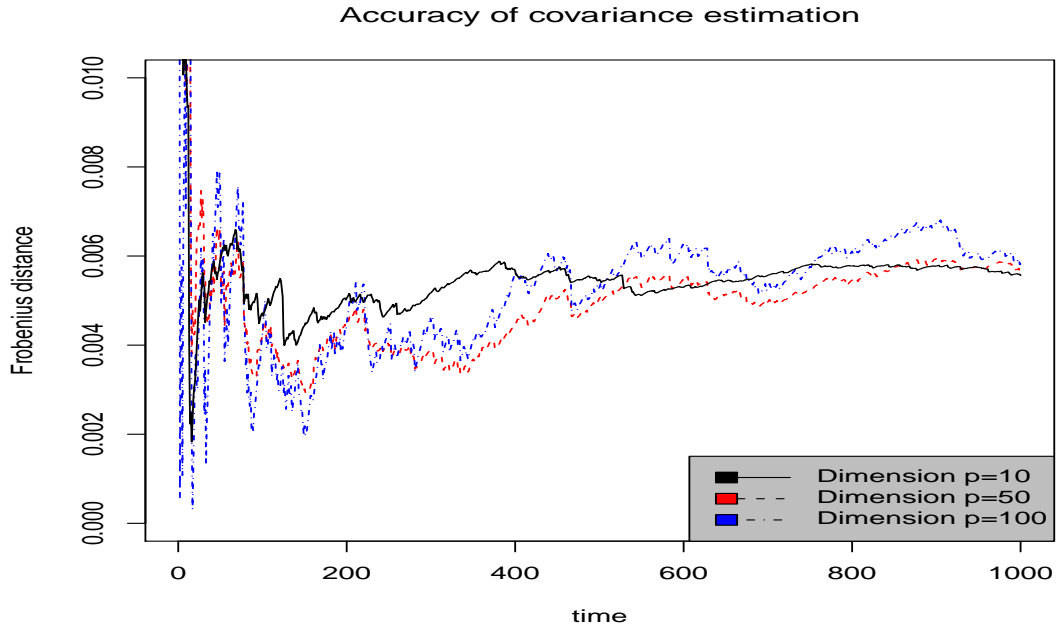, height=9cm, width=15cm}
 \caption{Frobenious distance over time of the GIW model using simulation from a single iteration. Shown are the three distances
 corresponding to simulations with $p=10,50,100$.}\label{fig1a}
\end{figure}

\subsection{Multivariate control charts}\label{controlcharts}

 In this section we consider a multivariate control charting scheme for autocorrelated data (Bersimis {\it et al.}, 2007). Typically multivariate control charts focus on the detection of signals of multivariate processes, which may exhibit out of control behaviour, defined as deviating from some prespecified target mean vector and a target covariance matrix. The Hotteling $T$ chart is the standard control chart as it is capable of detecting out of control signals of the joint effects of the variables of interest.

 \begin{figure}[t]
 \epsfig{file=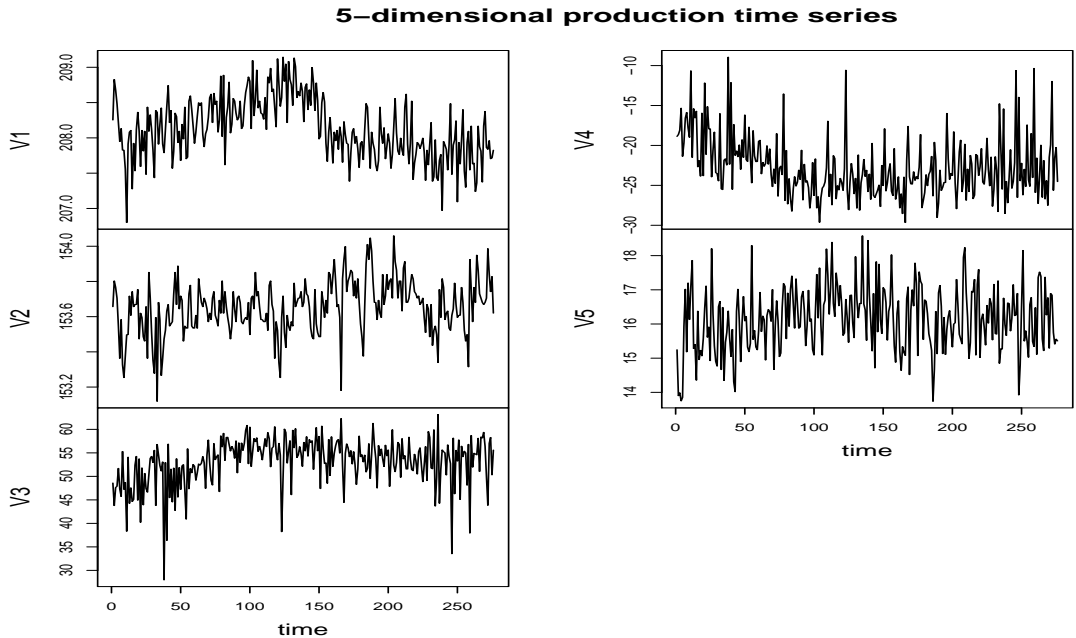, height=9cm, width=15cm}
 \caption{Production time series data.}\label{fig1}
\end{figure}

 However, many authors have pointed out that in the presence of autocorrelation, this chart is a poor performer (Vargas, 2003). As a result over the past decade researchers have focused considerable efforts on to the development of control charts for multivariate time series data (Bersimis {\it et al.}, 2007). Pan and Jarrett (2004) point out the importance of accurate estimation of the observation covariance matrix and they study the effects its miss-specification has in the detection of out of control signals. These authors suggest using the $T$ chart as above, after estimating the covariance matrix deploying some suitable time series method.

The multivariate local level model is a natural candidate model for the above situation, as it is a generalization of the popular Shewhart-Deming model, according to which the observed data $y_t$ are modelled as noisy versions of a constant level $\theta$, or $y_t=\theta+\epsilon_t$, where $\epsilon_t\sim N_p(0,\Sigma)$. This model is valid for serially uncorrelated data, but it is clearly not suitable for time series data. In this context, the motivation for the local level model is that the level of the time series at time $t$, $\theta_t$, follows a slow evolution described by a random walk. Using this model and considering an inverted Wishart distribution for $\Sigma$, Triantafyllopoulos (2006) proposes that first the one-step forecast distribution is sequentially produced, then the logarithm of the Bayes factors of the current forecast distribution against a prespecified target distribution forms a new univariate non-Gaussian time series, which control chart is designed using the modified exponentially weighted moving average (EWMA) control chart. If the process is on target, then the log Bayes factor (West and Harrison, 1997, \S11.4) will fluctuate around zero and the EWMA control chart will not signal significant deviations from this target. If, on the other hand, the EWMA signals out of control points, this will in turn signal deviations of the original process from its target. In the above reference, the target distribution is chosen to be a multivariate normal distribution, but, depending on experimentation and historical information, other distributions may be selected. As in Pan and Jarrett (2004) and in other studies, a critical stage in the application of this method, is that the estimate of $\Sigma$ and the forecast of $y_t$ are accurate, so that the fitted model is a good representation of the generating process.

We consider data from an experiment of production of a plastic mould the quality of which is
centered on the control of temperature and its variation. For this
purpose five measurements of the temperature of the mould have
been taken, for $276$ time points; for more details on the set up of this
experiment the reader is referred to Pan and Jarrett (2004). From Figure \ref{fig1}, which is a plot of the data, we can argue that this data possesses a local level type evolution. We have applied the above control charting methodology using the local level model with
the GIW distribution. For the model fit we note that the MSSE is $\textrm{MSSE}=[0.952~ 0.966~ 0.985~ 1.110~ 0.994]'$, which marks a much improved performance compared to Pan and Jarrett (2004) and to Triantafyllopoulos (2006); similar improved results (not shown here) apply considering other measures of goodness of fit, e.g. the mean of squared forecast errors and the mean absolute deviation. For the design of the control chart, with a small smoothing factor equal to 0.05 we use the EWMA chart, which control limits are modified from its usual control limits, to accommodate for both the non-Gaussianity of the Bayes factor series and its autocorrelation. Figure \ref{fig2} shows the EWMA control chart, from which we can see the improved behaviour: in Phase I where the model is applied and tested, we see that all EWMA points are within the control limits and in Phase II we see that the model signals a clear out of control behaviour. In contrast to the studies above, our model manages to avoid having out of control signals in Phase I, which reflects on the more accurate estimation of the observation covariance matrix and of the overall fit. In Phase II it shows a deterioration of the process, which is not signaled in Pan and Jarrett (2004) as very few out of control points are detected in that study. We also note that this deterioration can not be detected or suspected by either looking at the time series plot in Figure \ref{fig1} or performing univariate control charts to each of the individual series. For this data set, applying the $T$ control chart after estimating $\Sigma$ using our method and Pan and Jarrett (2004) again favoured our proposal (results not shown here). Finally we report that the improved performance of our chart in Phase I is evident, by noting that the control limits are much tighter as compared to those in Triantafyllopoulos (2006) and thus the deployed fitted model here, is a more accurate representation of the data.

\begin{figure}[t]
 \epsfig{file=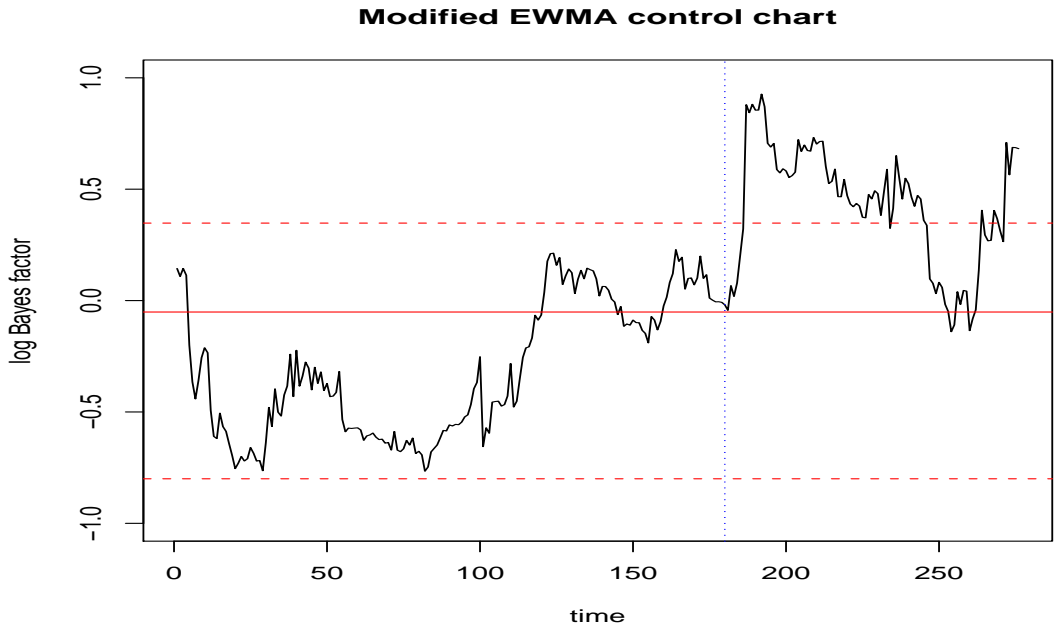, height=9cm, width=15cm}
 \caption{Modified EWMA control chart for the log Bayes factor of the Production process, using a smoothing
 parameter $\lambda=0.05$; the dotted horizontal
 line indicates the target mean $-0.141$ and the dashed horizontal lines indicate the
 control limits; the vertical line separates Phase I (for $t=1-180$) and Phase II
 (for $t=181-276$).}\label{fig2}
\end{figure}

\section{Conclusions}\label{discussion}

In this paper we propose on-line estimation for the multivariate local level model with the focus placed on the estimation of the covariance matrix of the innovations of the model. We criticize the application of the inverse Wishart prior distribution in this context as restrictive and often lacking empirical justification. Motivated from the conjugate model, we generalize the inverse Wishart distribution to account for wider application, but still manage to achieve approximate conjugacy, which is useful for real-time estimation. This approach results in fast recursive estimation, which resembles the Kalman filter, but allowing for covariance learning too. It is shown that our proposal delivers under Monte Carlo experiments and also in comparison with existing methods. An application of multivariate control charts is used to illustrate the proposed methodology. Future research efforts will be devoted on to the application of this methodology to high dimensional data.

\section*{Acknowledgements}

I am grateful to two anonymous referees for providing useful comments on an earlier draft of the paper.

\end{document}